\documentclass[onecolumn,letterpaper,journal]{IEEEtran}
\usepackage{upgreek}
\usepackage{mathtools, xparse}
\usepackage{graphicx}
\usepackage{epstopdf}
\usepackage{bm}
\usepackage{scalerel}
\usepackage{amsfonts,amsmath,amssymb}
\usepackage{amsthm}
\usepackage{algorithm2e,algorithmic}
\usepackage{mathtools}
\usepackage{multirow}
\usepackage{verbatim}
\usepackage{url}
\usepackage{comment}
\usepackage{mathtools}
\usepackage{epsf,pgf,graphicx}
\usepackage{color}
\usepackage{tikz,pgf}
\usepackage{diagbox}
\usepackage{makecell}
\usepackage{color}
\usepackage{braket}
\usepackage{mathtools}
\usepackage{float}
\usepackage{tikz}
\usetikzlibrary {positioning}
\usepackage{pgf}
\usepackage{qcircuit}
\usepackage{multirow}
\usepackage{listings}
\usepackage{array}
\usepackage{makecell}

\usepackage[switch]{lineno}

\SetCommentSty{mycommfont}

\SetKwInput{KwInput}{Input}                
\SetKwInput{KwOutput}{Output} 

\usepackage{graphicx,subcaption,cleveref}

\newcommand{\op}{\oplus}

\DeclarePairedDelimiter{\abs}{\lvert}{\rvert}

\newcommand{\pr}{\mathcal{P}}
\newcommand{\xv}{\mathbf{x}}
\newcommand{\yv}{\mathbf{y}}

\newcommand{\Af}{{\mathcal{A}^{(3,2)}}}
\newcommand{\As}{\mathcal{A}^{(3,3)}}

\newcommand{\Ak}[2]{\mathcal{A}^{(#1,#2)}}

\DeclarePairedDelimiter\kett{\lvert}{\rangle}

\newcommand{\var}[1]{\mathbf{#1}}

\newtheorem{fact}{Fact}
\usepackage{graphicx}

\newtheorem{proposition}{Proposition}
\newtheorem{corollary}{Corollary}
\newtheorem{theorem}{Theorem}
\newtheorem{lemma}{Lemma}
\newtheorem{definition}{Definition}
\newtheorem{remark}{Remark}

\begin{document}
\title{Following Forrelation -- 
Quantum Algorithms in Exploring Boolean Functions' Spectra}

\author{\IEEEauthorblockN{Suman Dutta,}
\IEEEauthorblockA{Indian Statistical Institute, Kolkata\\
Email: sumand.iiserb@gmail.com\\}
\and
\IEEEauthorblockN{Subhamoy Maitra,}
\IEEEauthorblockA{Indian Statistical Institute, Kolkata\\
Email: subho@isical.ac.in\\}
\and
\IEEEauthorblockN{Chandra Sekhar Mukherjee,}
\IEEEauthorblockA{Indian Statistical Institute, Kolkata\\
Email: chandrasekhar.mukherjee07@gmail.com}}

\maketitle

\begin{abstract}
Here we revisit the quantum algorithms for obtaining Forrelation [Aaronson et al, 2015] values
to evaluate some of the well-known cryptographically significant spectra of Boolean functions, 
namely the Walsh spectrum, the cross-correlation spectrum and the autocorrelation spectrum.
We introduce the existing 2-fold Forrelation formulation with 
bent duality based promise problems as desirable instantiations. 
Next we concentrate on the $3$-fold version through two approaches. First, we judiciously 
set-up some of the functions in $3$-fold Forrelation, so that given an oracle access, 
one can sample from the Walsh Spectrum of $f$. Using this, we obtain improved results than what 
we obtain from the Deutsch-Jozsa algorithm, and in turn it has implications in 
resiliency checking. Furthermore, we use similar idea to obtain 
a technique in estimating the cross-correlation (and thus autocorrelation) value at any point,
improving upon the existing algorithms. Finally, we tweak the quantum algorithm with 
superposition of linear functions to obtain a cross-correlation sampling technique. 
To the best of our knowledge, this is the first cross-correlation sampling algorithm with 
constant query complexity. This also provides a strategy to check if two functions are 
uncorrelated of degree $m$. We further modify this using Dicke states so that the time 
complexity reduces, particularly for constant values of $m$.
\end{abstract}

\begin{IEEEkeywords}
Boolean Functions, Cross-correlation Spectrum, Forrelation, Quantum Algorithm, Walsh Spectrum.
\end{IEEEkeywords}

\section{Introduction}
\label{sec:intro}
Quantum computing is one of the fundamental aspects in Quantum mechanics with many exciting possibilities. However, showing separation 
between classical and quantum paradigm still remains a very hard and interesting problem. In this regard, the black-box model is widely 
used to show the separation in respective domains. Some of the most well known quantum algorithms, such as Shor's factoring~\cite{shor}, 
Grover's search~\cite{grover} and Simon's hidden shift~\cite{simon} are all defined in this black-box paradigm. Naturally, this black-box model 
(also known as the query model) is one of the most studied domains of quantum computer science. In this domain, problems can be studied 
in the exact quantum and bounded error quantum model as well as probabilistic and deterministic classical model.

One of the simplest yet fundamental algorithms in this domain is the Deutsch-Jozsa (DJ) algorithm~\cite{deutsch}. Informally speaking, 
given access to a Boolean function $f$, it creates the Walsh spectrum in the form of a superposition, where the amplitudes of the final 
states are the normalized Walsh spectrum value of the function at the respective points. If $f$ is promised to be either constant or balanced, 
the DJ algorithm deterministically concludes which one it is with a single quantum query, whereas any deterministic classical algorithm 
requires exponentially many queries, which is asymptotically the maximum possible separation between these two models. However obtaining the 
separation between the probabilistic classical model and the bounded error quantum model is not as simple and recently resolved in \cite{bansal}.

In this direction the ``Forrelation" formulation is of great interest. This concept was first introduced in the work by Aaronson 
et al~\cite{forr} and then was used in the seminal paper by Aaronson et al~\cite{aaron}. They showed that the Forrelation problem has 
constant versus exponential (in $n$) query complexity separation in the bounded error quantum and the probabilistic classical model. 
Recently this concept was modified by Tal et al~\cite{tal} to obtain even higher separation between these two. It is quite easy to see that 
the Forrelation problem is efficiently solvable in the quantum paradigm and the main contribution and focus of those papers were towards 
showing that no probabilistic classical algorithm could solve this problem efficiently. In this backdrop, before proceeding further, 
let us start by introducing the $k$-fold Forrelation set-up.

The $k$-fold Forrelation of $k$ Boolean functions 
$f_1, \ldots, f_k$, each defined on $n$ variables can be expressed as 
\begin{align*}
&\Phi_{f_1,\ldots , f_k} =\frac{1}{2^{(k+1)n/2}}\sum\limits_{\var{x_1}, \ldots , \var{x_k} \in \{0,1\}^n}
f_1(\var{x_1})(-1)^{\var{x_1} \cdot \var{x_2}} f_2(\var{x_2})  \ldots (-1)^{\var{x_{k-1}} \cdot \var{x_k}}  
f_k(\var{x_k}).
\end{align*}
Specifically, Aaronson et al~\cite{aaron} showed that the problem of identifying whether 
$\abs{\Phi_{f_1,\ldots ,f_k}}\leq\frac{1}{100}$ or $\Phi_{f_1,\ldots ,f_k}\geq\frac{3}{5}$ has 
$\Omega \left(\frac{2^{\frac{n}{2}}}{n} \right)$ query complexity in the probabilistic classical model as opposed to 
constant query complexity in the bounded error quantum model. To prove those results, they designed the following two quantum algorithms 
assuming oracle access to $f_1$ through $f_k$. 
\begin{enumerate}
\item $\Ak{k}{k}(f_1, \ldots f_k)$: It outputs the state $0^n$ with probability 
$\Phi_{f_1,\ldots , f_k}^2$ upon measuring the $n$ predetermined qubits.

\item $\Ak{k}{\lceil \frac{k}{2} \rceil}(f_1, \ldots f_k)$:
It outputs the state $0$ with probability 
$\frac{1+\Phi_{f_1,\ldots , f_k}}{2}$ upon measuring $1$ predetermined qubit.
\end{enumerate}
In this paper, we analyze the Forrelation problem from a different perspective, in terms of
studying and evaluating different spectra of Boolean functions in the black-box (query) model. 
While $k$-fold Forrelation can be defined for any $k$ functions on $n$ variables, 
we concentrate mostly on the cases $k = 2, 3$ in this study. 

Boolean functions are one the most fundamental combinatorial objects in the domain of computer science, 
with applications in coding theory, cryptology as well as in understanding power and limitations of different 
computational models. The spectra that we consider here have many applications and in particular in the domain of cryptology.
When a Boolean function is used as a primitive in a cryptosystem, its quality is studied by different parameters related 
to the spectra based on certain transformations~\cite{bera,car93,chak,XZ88,cross,ps04}.
In this paper, we sometimes refer to Boolean functions simply as functions. 

\subsection{Organization and Contribution}
The preliminaries and the background materials are presented in Section~\ref{sec:2}. Next a warm-up section
(Section~\ref{sec:3}) is oriented around observing how certain promise problems
can be reduced to different desirable instantiations of the $2$-fold Forrelation,
where the Forrelation of two functions $f$ and $g$ is denoted by $\Phi_{f,g}$.
We start by studying the scenario when $\Phi_{f,g}$ assumes the maximum value $1$ and the scenario when $\Phi_{f,g}$ is close to $0$. Based on them, we tailor several bent-duality based promise problems. 
We also observe that in certain situations the simple DJ algorithm works as efficiently as 
the Forrelation set-up, while in the other (majority of the) cases there is no obvious way of achieving the same using the DJ algorithm. 
In summary, we note that the techniques in solving the Forrelation problem extend the idea of the DJ algorithm much further. 
We like to point out that in~\cite{roet}, certain quantum algorithms related to bent functions have been studied, though 
that is not related to the problems we consider here.

In the next part (Section~\ref{sec:4}) we concentrate on the $3$-fold Forrelation formulation
which we use as a unifying framework for evaluating different spectra of functions, 
and this is where the most interesting results lie. 
Our initial approach in this direction is driven by an observation about
the $3$-fold Forrelation formulation. 
As we shall obtain in Section~\ref{sec:4}, the $3$-fold Forrelation for three 
functions $f_1,f_2$ and $g$ can be written as 
$\Phi_{f_1,g,f_2}=\frac{1}{2^{3n/2}}\sum_{\xv \in \{0,1\}^n}g(\xv)W_{f_1}(\xv)W_{f_2}(\xv)$ where the functions 
are defined from $ \{0,1\}^n$ to $\{1,-1\}$ 
(note that the definition deviates from the usual realization of a Boolean function as $f: \{0,1\}^n \rightarrow \{0,1\}$,
however this is a simple linear transformation of the output bits of the form $0 \rightarrow 1$ and 
$1 \rightarrow -1$). Further, note that the $3$-fold Forrelation algorithms $\Af$ and $\As$ 
need oracle access to all the three functions $f_1, f_2$ and $f_3$. 

Against this backdrop, we set $f_1=f_3=f$ and cleverly design a function $g$
that we set as $f_2$ to estimate the Walsh spectrum value $\left(W_f(\var{u})\right)$ 
of a function at any given set of points $S=\{\var{u}\} \subseteq \{0,1\}^n$. Here we set 
$g$ to be a function such that $g(\xv)=-1$ if and only if $\xv\in S$ and $1$, otherwise.
Then we observe that the sampling probability obtained from using the algorithm $\Af(f,g,f)$ is equivalent to that of the Deutsch-Jozsa algorithm. However, using the algorithm $\As(f,g,f)$, we obtain a constant factor improvement over the DJ algorithm.
The problem of checking whether a given Boolean function $f$ is $m$-resilient or not, can also be improved upon using $\As(f,g,f)$ where $g$ is designed such that $g(\xv)=-1,\,\,\forall \xv:wt\left(\xv\right)\leq m$ and $1$, otherwise. Here $wt(\xv)$ denotes the count of $1$'s in the bit pattern $\xv$.
This provides a constant advantage in the query complexity compared to~\cite{chak} for checking whether a function $f$ is $m$-resilient or not. 
Although the improvement is only constant, but the initial sampling algorithm outperforms Deutsch-Jozsa when compared with respect to 
the same number of queries. This is an interesting result and underlines the ability of the Forrelation formulation to provide further 
insight about such problems.

Next, in Section~\ref{sec:5}, we study the cross-correlation spectrum $C_{f,g}(\xv)$ of two Boolean functions $f$ and $g$ and connect it to the $3$-fold 
Forrelation set-up. First we design an algorithm to estimate the cross-correlation value at any particular
point in the same way as resiliency checking. We use two very interesting results to link the two formulations, cross-correlation and Forrelation. 
Consider the $2^n \times 2^n$ Hadamard matrix 
$\hat{H}_n=
{\begin{pmatrix}
1 & 1\\
1 & -1 
\end{pmatrix}}^{\otimes n}
$. 
One may note that (see~\cite[Theorem 3.1]{cross})
$$[C_{f,g}(000\ldots 0), \ldots, C_{f,g}(111\ldots 1) ] H_n = 
[ 
W_fW_g(000...0), \ldots, W_fW_g(1111 \ldots 1)
].
$$
For the ``non-normalized'' Hadamard matrix we have $\left( \hat{H}_n \right)^2=2^nI_n$, where $I_n$ is the $2^n$-dimensional identity matrix. 
We use this along with the fact that the rows of $H_n$
has a one-to-one correspondence with the $2^n$ linear Boolean functions on $n$ variables. 
Thus we can obtain different values of 
$\sum_{\xv \in \{0,1\}^n}g(\xv)W_{f_1}(\xv)W_{f_2}(\xv)$ where $g$ is a chosen linear function to obtain the cross-correlation value 
at different points. This leads us to an algorithm $A(f,g,\var{u})$ so that for any two functions $f, g$ on $n$ variables and for any point 
$\var{u} \in \{0,1\}^n$, the probability of ``obtaining an all zero output while measuring the $n$ predetermined qubits" is 
$\frac{\left(C_{f,g}(\var{u})\right)^2}{2^{2n}}$. This provides a technique to estimate the cross-correlation spectrum. As a corollary, 
we obtain an autocorrelation spectrum estimation algorithm. We compare the later with the autocorrelation estimation technique in~\cite{bera}
and observe certain improvements. 

Finally we move towards a cross-correlation sampling algorithm. To design this, we tweak the algorithm $\As$ for $3$-fold Forrelation. 
Instead of considering a linear function in place of $f_2$, we implement superposition of all linear functions using an additional $n$ qubits.
This methodology allows us to sample from the cross-correlation spectrum of two functions $f$ and $g$. As the final contribution, we study 
the problem of checking whether two functions $f$ and $g$ are uncorrelated of degree $m$, which is to check
if $C_{f,g}(\var{u})=0$ for all $\var{u}$ of weight less than or equal to $m$. In this direction we use Dicke states to further modify 
the algorithm to improve the query complexity compared to what would have been required if we used
the cross-correlation sampling algorithm. To the best of our knowledge, such a sampling algorithm for cross-correlation spectrum was not 
known before. Moreover this allows to view Forrelation as a unifying framework through which we are able to study these different 
spectra of Boolean functions.

To elaborate further, we first obtain a set-up so that the probability of getting particular outputs
in $\As$ and $\Af$ are $\frac{\left(C_{f,g}(\var{u})\right)^2}{2^{2n}}$ and
$\frac{1}{2}\left(1+\frac{C_{f,g}(\var{u})}{2^n}\right)$ respectively.
We obtain the results by setting a linear function of our choice $\mathbb{L}_{\var{u}}$ as $f_2$ in
the $3$-fold Forrelation formulation, and then setting $f_1=f$ and $f_3=g$.
This gives us two cross-correlation estimation algorithms and the
autocorrelation estimation algorithms as immediate corollaries. We observe that the algorithm due to $\Af$
fairs better than the autocorrelation estimation algorithm due to~\cite{bera}
for comparable accuracy. The improvements are detailed in Section~\ref{compbera}.

We then tweak $\As(f,\mathbb{L}_{\var{u}},g)$
to obtain a cross-correlation sampling algorithm. This is obtained by replacing
the linear function (placeholder of $f_2$) by a superposition of linear functions.
The algorithm results the output $\var{u}||0^n$ with probability
$\frac{\left(C_{f,g}\right)^2}{2^{3n}}$. To the best of our knowledge, this is the first
cross-correlation sampling algorithm with constant query complexity. In fact,
it only makes one query each to the functions $f$ and $g$.

Finally we study the problem of checking if two functions
are uncorrelated of degree $m$, that is $C_{f,g}(\var{u})=0$ for all $\var{u}$
of weight $m$ or less. Note that, this is similar to the resiliency checking
problem that we have studied here, and a similar approach of applying
amplitude amplification on the cross-correlation sampling algorithm guarantees
correct output with constant probability with a query complexity of
$\mathcal{O}(\frac{1}{a})$ where
$a^2=\frac{1}{2^{3n}} \sum_{\xv: wt(\xv) \leq m}
C_{f,g}(\xv)^2$. As the final contribution of this paper,
we modify this algorithm to improve the query complexity, to
$\mathcal{O}(\sum_{i=0}^m\frac{1}{b_i})$ queries to $f$ and $g$ each,
where
$(b_i)^2= \frac{1}{{n \choose i}2^{2n}} \sum_{\xv: wt(\xv) = i} C_{f,g}(\xv)^2$.

In summary, the main contribution of this paper is in studying different fundamental (as well as cryptographically significant)
spectra of Boolean functions using quantum algorithms related to Forrelation. In the process, the results we obtain are new or the superior 
ones than the existing results~\cite{bera,chak}. 
We conclude this paper with future directions in Section~\ref{sec:6}. As a passing remark, all these algorithms are implemented for small values of $n$ 
(generally 4) in the IBMQ simulator and subsequently the results are verified.

\section{Preliminaries and Background}
\label{sec:2}
Let us start with some of the basic definitions, notations and notions that we will be using frequently throughout this paper.
\subsection{Some Properties of Boolean function}
\begin{definition}
\label{def:bf}
Let $\mathbb{F}_2$ be the prime field of characteristic $2$ and 
$\mathbb{F}^n_2\equiv \{\var{x}=\left(x_1,x_2,\ldots, x_n \right):x_i\in\mathbb{F}_2,1\leq i \leq n \}$ 
be the vector space of dimension $n$ over $\mathbb{F}_2$. An $n$-variable Boolean function is a mapping from 
$\mathbb{F}^n_2$ to $\mathbb{F}_2$. In this paper we consider the Boolean functions $f$ such that $f:\{0,1\}^n\rightarrow \{-1,1\}$.

For a given Boolean function $f(\var{x})$, the Walsh transform of $f$ is an integer valued function, 
$W_f:\{0,1\}^n \rightarrow [-2^n,2^n]$ which is defined as 
$$W_f(\bm{\omega})=\sum_{\var{x}\in\{0,1\}^n}{f(\var{x})\cdot(-1)^{\var{x}\cdot \bm{\omega}}},$$
where $\var{x}\cdot \bm{\omega} = x_1\omega_1\oplus x_2\omega_2\oplus \ldots \op x_n \omega_n$ is called the inner product of 
$\var{x}$ and $\bm{\omega}$.
\end{definition}
One can observe that the set of all Walsh spectra corresponding to the $2^{2^n}$ 
functions for any given $n$, forms a very small subset of the complete
integer valued functions $F: \{0,1\}^n \rightarrow [-2^n, 2^n]$.
Furthermore, for any function on $n$ variables, we have the constraint that
$\sum_{ \xv \in \{0,1\}^n} \left(W_f(\xv)\right)^2=2^{2n}$, 
which is the well known Parseval's identity.

In this regard, we define the bent functions (see~\cite{dil74,car93} for more details) for which the absolute value of the Walsh 
transforms are equal for all the inputs.
These functions are the maximally distant functions from the set of all linear functions 
and are widely used as cryptographic primitives (with certain modifications and compositions, as bent functions are not balanced). 
The bent functions also have the unique property of duality, where a bent function effectively mimics the Walsh spectrum of 
another bent function. Let us now formally define bent functions and the concept of 
duality of bent functions.

\begin{definition}
\label{def:bent}
A bent function $f:\{0,1\}^n\rightarrow \{1,-1\}$ is a Boolean function (where $n$ is even), such that the Walsh 
spectrum $W_f(\bm\omega)=\pm 2^{\frac{n}{2}}$ for all $\bm\omega\in\{0,1\}^n$.

Such functions are of maximum nonlinearity for even $n$.
Two bent functions $f,\hat{f}:\{0,1\}^n\rightarrow \{1,-1\}$ are called dual of each other 
if for all $\bm\omega\in\{0,1\}^n$, we have 
$W_f(\bm\omega)=2^{n/2} \hat{f}(\bm\omega)$ and $W_{\hat{f}}(\bm\omega)=2^{n/2} f(\bm\omega)$. 
That is $f(\omega)= sign(W_{\hat{f}}(\omega))$ and vice versa where $sign(a)=1$ if $a>0$ and $-1$ otherwise. 
Furthermore, if $f=\hat{f}$ then $f$ is called its self dual.
\end{definition}
Let us now define a class of bent functions along with some interesting properties that we use in this paper. This class of bent functions comes from the definition of Kerdock codes. To know more about Kerdock codes, one may refer to \cite{kerdock,kerdock1} and the references therein. 
For any two bent functions $f_1$ and $f_2$ arising from the Kerdock code, $f_1 \op f_2$ is also a bent function.

Then we define another property of the Boolean functions, called the resiliency (see~\cite{ps04} for more details and cryptographic implications). 
Let us first define the (Hamming) weight of $\xv \in \{0,1\}^n$. It is simply the number of 
$1$'s in the bit pattern $\xv$. Then the resiliency of a function $f$ is defined as follows~\cite{XZ88}. 
\begin{definition}
An $n$-variable Boolean function $f$ is called $m$-resilient (for $m<n$) if and only if the Walsh transforms, 
$W_f(\bm{\omega})=0$ for $0\leq wt(\bm{\omega})\leq m$.

That is, if there exists an $\bm\omega\in\{0,1\}^n$ such that $wt(\bm{\omega})\leq m$ and $W_f(\bm{\omega})\neq 0$, then $f$ is not $m$-resilient.
\end{definition}
Finally we describe the concepts of cross-correlation spectra of two functions $f$ and $g$ and 
the autocorrelation spectra of any function $f$, which have a wide range of cryptographic importance (see~\cite{cross} for more details).
\begin{definition}
The cross-correlation of two functions $f,g: \{0,1\}^n \rightarrow \{1,-1\}$
at a point $\var{y} \in \{0,1\}^n$ is defined as 
$$C_{f,g}(\var{y})=\sum\limits_{\xv \in \{0,1\}^n} f(\xv)g(\xv \op \var{y}).$$
Similarly, the autocorrelation of a function $f$ at a point
$\var{y} \in \{0,1\}^n$ is defined as 
$C_{f}(\var{y})=\sum\limits_{\xv \in \{0,1\}^n} f(\xv)f(\xv \op \var{y})$.
\end{definition}
We also refer to the following terminology.
\begin{definition}
\label{def:uncor}
Two functions are called to be uncorrelated of degree $k$ if $C_{f,g}(\var{y})=0$ 
for all $\var{y}: 0 \leq wt(\var{y}) \leq k$. 
\end{definition}
It is desirable that the component functions of a cryptographic system 
are pairwise as much uncorrelated as possible. 

\subsection{Quantum algorithms of interest}
We refer to~\cite{nc} for basics of quantum information and computing. In this paradigm, let us now describe the working of 
an algorithm in the black-box (query) model.\\

\noindent{\bf Black-box model}
In this model oracle access to $U_f$ is given for a function $f$ on $n$ variables whose 
internal structure is unknown. We first define the $n$ qubit state $\ket{\xv}, \xv \in \{0,1\}^n$
defined as $\ket{\xv}=\otimes_{i=1}^n \ket{x_i}$. These qubits form the query registers.
In this setting, the functioning of $U_f$ on the 
$n+1$ qubit state $\ket{\xv}\ket{a}$ (where $\ket{a}$ is used for storing the output) 
is defined as $U_f\ket{\xv}\ket{a}=\ket{\xv}\kett*{a \oplus \frac{1-f(\xv)}{2}}$
(this is to accommodate the fact that the range of $f$ is $\{1,-1\}$ and not $\{0,1\}$). 
Then if we set $\ket{a}=\ket{-}$, we have $U_f\ket{\xv}\ket{-}= f(\xv)\ket{\xv}\ket{-}$.
This is also called the phase-kickback method. 
Now we describe the Deutsch-Jozsa (DJ) algorithm,
which is one of the most fundamental algorithms in this model. \\

\noindent{\bf Deutsch-Jozsa algorithm}
Let there be a function $f$ on $n$ variables.
Then starting from the state $\ket{\psi}=\ket{0}^{\otimes n}\ket{-}$,
the working of the DJ algorithm is $(H^{\otimes n}\otimes I)U_f(H^{\otimes n}\otimes I)\ket{\psi}$
where $H^{\otimes n}$ is the $n$ Hadamard gates, applied to the $n$ query registers.
This results in the state
$\frac{1}{2^n}\sum_{\xv \in \{0,1\}^n} W_f(\xv) \ket{\xv}.$
The promise problems such as whether a function is 
balanced or constant, can be resolved using this algorithm. 
The basic criteria for deterministically determining the promise problems using the
DJ algorithm is observing that for all the problems there is a state $\xv \in \{0,1\}^n$ 
(known to the user based on the promise) such that for the first promise $W_f(\xv)=2^n$
and for the second promise $W_f(\xv)=0$, which allows us to distinguish between the 
two cases with certainty.
Next we study the Forrelation problem.\\

\noindent{\bf The Forrelation problem}
The Forrelation property is one of the central results in the study of separating 
the computational power of the bounded error quantum and classical probabilistic models
in the black-box set up. Forrelation is defined as follows. 

\begin{definition}[\cite{forr}]
The $2$-fold Forrelation of two functions $f$ and $g$ defined on $n$ variables is 
the measure of correlation between the function $f$ and the Walsh transform of $g$.
Formally, the $2$-fold Forrelation is expressed as
$$\Phi_{f,g}=\frac{1}{2^{3n/2}} \displaystyle \sum_{\var{x} \in \{0,1\}^n} f(\var{x})W_g(\var{x})=\frac{1}{2^{3n/2}}\sum\limits_{\var{x_1},\var{x_2} \in \{0,1\}^n}
f_1(\var{x_1})(-1)^{\var{x_1} \cdot \var{x_2}} f_2(\var{x_2}).$$

This can be then extended to any $k>2$, where the $k$-fold Forrelation for $k$ functions 
$f_1, \ldots , f_k$ each defined on $n$ variables can be expressed as 
$$\Phi_{f_1,\ldots , f_k} =\frac{1}{2^{(k+1)n/2}}\sum\limits_{\var{x_1}, \ldots , \var{x_k} \in \{0,1\}^n}
f_1(\var{x_1})(-1)^{\var{x_1} \cdot \var{x_2}} f_2(\var{x_2})  \ldots (-1)^{\var{x_{k-1}} \cdot \var{x_k}}  
f_k(\var{x_k}).$$
\end{definition}

In this regard, we have the following theorem due to \cite{aaron}.
\begin{theorem}[\cite{aaron}]
\label{th:aaron}
The problem of determining $\abs*{\Phi_{f,g}} \leq \frac{1}{100}$ or $\Phi_{f,g} \geq \frac{3}{5}$ (under the promise that it is one of the two cases) can be solved in the bounded error quantum query model with a single $query$ whereas any randomized classical algorithm would require at least $\Omega \left(\frac{2^{n/2}}{n} \right)$ queries in the worst case.
\end{theorem}

The work in~\cite{aaron} designed two efficient quantum algorithms for calculating the $k$-fold Forrelation
of functions, one using $k$ queries and the other using $\lceil  \frac{k}{2} \rceil$
queries.
Let us now briefly discuss the two algorithms and their final outputs as mentioned in~\cite{aaron}.
In both the algorithms we assume oracle access to the functions $f_1, f_2, \ldots ,f_k$,
consistent with the black-box model.

\begin{enumerate}
\item {\bf The $k$-query algorithm:}
The algorithm starts with the state $\ket{\psi}=\ket{0}^{\otimes n}\ket{-}$ and following the steps $H^{\otimes n} \rightarrow U_{f_1} \rightarrow H^{\otimes n} \rightarrow U_{f_2} \rightarrow H^{\otimes n} \rightarrow  \ldots  \rightarrow U_{f_k} \rightarrow H^{\otimes n}$
produces the final state $\ket{\psi_{end}}$, where $H^{\otimes n}$ denotes $n$-copies of Hadamard gates applied to the first $n$ qubits.
Here it is easy to verify that in the state $\ket{\psi_{end}}$ (ignoring the last qubit
which always remains in the $\ket{-}$ state) the amplitude of 
the state $\ket{0}^{\otimes n}$ is $\Phi_{f_1, \ldots , f_k}$. 
That is, if we measure the $n$ qubits, the output will be all zero with probability
$(\Phi_{f_1, \ldots , f_k})^2$.

\item {\bf The $\lceil \frac{k}{2} \rceil$-query algorithm:}
In this algorithm we consider one more qubit, initialized to $\ket{+}$ state, known as the driving qubit.
Thus the starting state is $\ket{+}\ket{0}^{\otimes n}\ket{-}$. 
Now controlled on the driving qubit being in the state $\ket{0}$
the following series of operations are applied to the next $n+1$ qubits: $H^{\otimes n}\rightarrow U_{f_1}\rightarrow H^{\otimes n} \rightarrow \ldots \rightarrow U_{f_{\lceil \frac{k}{2} \rceil}}\rightarrow H^{\otimes n}$, where $H^{\otimes n}$ is applied to the $n$ qubits starting from $2$ to $n+1$.
Similarly, controlled on the driving qubit being on the state $\ket{1}$,
the following operations are applied to the rest of the $n+1$ qubits in the given order: $H^{\otimes n}\rightarrow  U_{f_k}\rightarrow  \ldots \rightarrow  H^{\otimes n}\rightarrow  U_{f_{\lceil \frac{k}{2} \rceil+1}}.$
Finally, the driving qubit is measured in the Hadamard basis, which is equivalent to applying a Hadamard gate and then measuring the driving qubit in the computational basis. 
The probability of observing the output as $\ket{0}$ is given by $\frac{1 + \Phi_{f_1, \ldots , f_k}}{2}$.
\end{enumerate}

In this paper we concentrate on the $2$-fold and the $3$-fold Forrelation problems. For $k=3$, we denote the $\lceil \frac{k}{2} \rceil$-query and the $k$-query algorithms
by $\Af$ and $\As$ respectively.\\

\noindent{\bf The non-resiliency checking algorithm of~\cite{chak}}
\label{amp-chak}
Let us now briefly describe the work by Chakraborty et al~\cite{chak} that 
probabilistically determines if a function $f$ is $m$-resilient, given oracle access $U_f$.
First the DJ algorithm produces the state 
$\sum_{\xv \in \{0,1\}^n} \frac{W_f(\xv)}{2^n} \ket{\xv}$.
Here observe that if $f$ is indeed $m$-resilient then $W_f(\xv)=0$ for all  $\xv\in\{0,1\}^n$ such that $wt(\xv) \leq m$.
Thus if we measure the state we will never observe a state with weight less than or equal to $m$. 
On the other hand if one obtains an output of weight less than or equal to $ m$ then it can be deterministically 
concluded that the function is not $m$-resilient. However if one does not get any 
$\xv: wt(\xv) \leq m$ output even after some polynomial number of trials, it still 
can not be claimed deterministically that $f$ is indeed $m$-resilient, as the $W_f(\xv): wt(\xv) \leq m$ 
may be non zero for only a very small number of $\xv$ and even then the Walsh Spectrum 
values maybe very low at those points.
In this regard, to obtain advantage over classical sampling, amplitude amplification was applied for all the sates with weight less than or equal to $m$.\\

\noindent{\bf The autocorrelation sampling and estimation algorithms of~\cite{bera}}
The paper~\cite{bera} first designs an algorithm to sample autocorrelation spectrum 
of any function $f$ and then designs another algorithm to estimate the autocorrelation spectrum 
at any point $C_f(\xv)$. We note down the output state of their algorithms 
and a general estimation algorithm that they use, to understand how the results due 
to the Forrelation formulation varies from those algorithms.

In this paper~\cite{bera}, Bera et al studied the higher order derivatives of the Walsh spectrum of a function using Deutsch-Jozsa and the equivalence between first order derivative and the autocorrelation spectrum
obtained as an output state to sample autocorrelation spectrum for any function in the following manner.

\begin{theorem}[\cite{bera}]
\label{bera:auto}
There is an algorithm $A_1$ such that $A_1\ket{0}^{\otimes 2n+1}$
produces the state 

$$\ket{\phi}= \ket{-}\frac{1}{\sqrt{2^n}} 
\sum \limits_{\var{b} \in \{0,1\}^n} \sum\limits_{\yv \in \{0,1\}^n} \widehat{\Delta f_{\var{b}}^{(1)}}(\yv)\ket{\yv}\ket{\var{b}}.
$$ 
where $\widehat{\Delta f_{0^n}^{(1)}}(\yv)=\frac{C_f(\yv)}{2^n}$.

\end{theorem}

Thus the probability of getting the state $\ket{\yv}\ket{0}^{\otimes n}$ is $\frac{C_f(\yv)^2}{2^{3n}}$.
Next the authors design a swap test based algorithm to lay the premise for a better estimation algorithm 
for $C_f(\yv)$.

\begin{lemma}[\cite{bera}]
\label{bera:auto-est}
There is an algorithm $A_2(\yv)$ defined on $3n+2$ qubits 
so that upon acting on $\ket{0}^{\otimes 3n+2}$ the probability of 
obtaining the output $0$ on measuring a predetermined qubit in 
the computational basis is given by $\frac{1}{2}\left(1+ \frac{C_f(u)^2}{2^{2n}}\right)$.
\end{lemma}

Finally they use a general quantum estimation algorithm that is used together with Lemma~\ref{bera:auto-est}
to obtain efficient sampling.

\begin{lemma}[\cite{brass}]
\label{bera:est}
Let $A$ be a quantum circuit without any measurement and let $p$ denotes the probability of observing its output state in a particular subspace. There is a quantum algorithm that makes a total of $\mathcal{O}(\frac{\pi}{\epsilon}\log \frac{1}{\delta})$ many (controlled) calls to $A$ and returns an estimate
$\tilde{p}$ such that $\pr[\tilde{p}-\epsilon \leq p \leq \tilde{p}+\epsilon ] \geq 1-\delta$
for any accuracy $\epsilon \leq \frac{1}{4}$ and error $\delta < 1$.
\end{lemma}

Combining the Lemmas~\ref{bera:auto-est} and~\ref{bera:est} the authors of~\cite{bera}
obtain the following result.

\begin{theorem}[\cite{bera}]
\label{th:final}
Given a function $f$ there exists an algorithm that makes 
$\mathcal{O}\left( \frac{\pi}{\epsilon} \log \frac{1}{\delta} \right)$ queries
and returns an estimate $\alpha$ such that
$\pr[\alpha-\epsilon \leq \frac{C_f(\yv)^2}{2^{2n}} \leq \alpha+\epsilon ] \geq 1-\delta$.
\end{theorem}

Here one should note that the algorithm corresponding to Lemma~\ref{bera:auto-est} is similar to the algorithm
$\Af$ for $3$-fold Forrelation in terms of using a driving qubit controlled on which the different
oracles are applied. However, as we shall observe that we are able to instantiate $\Af$ in such a manner
that gives us a more efficient estimation algorithm when we compare the corresponding accuracy values.
Let us now move towards the study of $2$-fold Forrelation.

\section{Warming up: $2$-fold Forrelation and bent Duality}
\label{sec:3}
Identifying one of the two given Boolean functions in the black box model has been a problem of interest in the last three decades. Consequently this has given rise to some of the most fundamental algorithms in the quantum oracle model, such as the Deutsch-Jozsa (DJ) algorithm. 
The DJ algorithm deterministically identifies a constant function from a balanced function in the quantum paradigm with a single query, whereas any classical deterministic algorithm would need exponential operations in terms of the number of queries. 

However, separating the bounded error quantum model and probabilistic classical model is much more challenging. One of the central results in this direction is the work by~\cite{aaron} which builds upon the Forrelation problem defined by Aaronson~\cite{forr} and studies the correlation of a Boolean function $f_1$ with the Walsh Spectrum of another Boolean function $f_2$. 
The Forrelation of two functions $f_1,f_2:\{0,1\}^n\rightarrow\{1,-1\}$ is formally defined as 
$$\Phi_{f_1,f_2}=\frac{1}{2^{3n/2}} \displaystyle \sum_{\var{x_1}\in \{0,1\}^n}
f_1(\var{x_1})W_{f_2}(\var{x_1})=\frac{1}{2^{3n/2}} \displaystyle \sum_{\var{x_1}, \var{x_2} \in \{0,1\}^n}
f_1(\var{x_1})(-1)^{\var{x_1} \cdot \var{x_2}}f_2(\var{x_2}).$$
Given two Boolean function, $f_1,f_2:\{0,1\}^n\rightarrow\{-1,1\}$ it is easy to see that $\Phi_{f_1,f_2}=\Phi_{f_2,f_1}$ and $-1 \leq \Phi_{f_1,f_2} \leq 1$. 
 
Given oracle access to the functions $f_1,f_2$, the algorithm begins with the state $\ket{0}^{\otimes n}\ket{-}$ and traverses through the following sequence of steps: $H^{\otimes n}\rightarrow U_{f_1}\rightarrow H^{\otimes n}\rightarrow  U_{f_{2}}\rightarrow
H^{\otimes n}$. Finally, after ignoring the last qubit, the amplitude of the all zero state becomes $\Phi_{f_1,f_2}$ and thus upon measurement 
the probability of obtaining the all zero string, $\ket{0}^{\otimes n}$ is given by
$\left( \Phi_{f_1,f_2} \right)^2$ \cite{aaron}.
The quantum structure of the algorithm is depicted in Figure \ref{fig:q2f2}.

\begin{figure}[ht]
\centering
\includegraphics[scale=0.9]{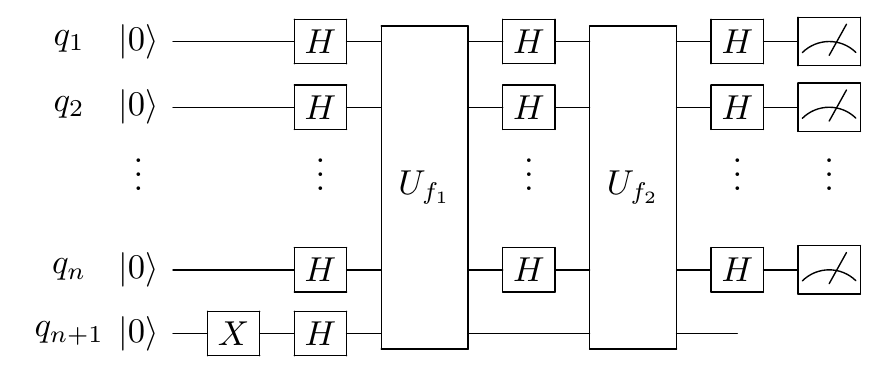}
\vspace{0.1cm}
\caption{Quantum circuit for implementing the $2$-fold Forrelation problem using $2$ queries.}
\label{fig:q2f2}
\end{figure}

We now present certain desirable instantiations of the Forrelation problem (as directed in Theorem~\ref{th:aaron})
by connecting different properties of Boolean functions. 
Moreover, we observe how this algorithm varies from the Deutsch-Jozsa algorithm and becomes similar to it in results depending on the choice of the functions $f_1$ and $f_2$. 

Let us now concentrate on the scenario where the Forrelation of two functions 
$f_1$ and $f_2$ on $n$ variables is maximum, i.e. $1$.
It is easy to see that this scenario happens when $n$ is even and $f_1$ and $f_2$ are 
both bent functions and are dual to each other, which we present in the following lemma.
In this section, unless otherwise mentioned, 
by $\hat{f}$, we denote the dual of $f$, when $f$ is a bent function.
\begin{proposition}
\label{prop:1}
Let $f,g:\{0,1\}^n\rightarrow \{1,-1\}$ be bent. Then $\Phi_{f,g} = 1$ if and only if $f$ and $g$ are dual to each other.
\end{proposition} 
\begin{proof}
If $g=\hat{f}$, from the definition of Forrelation, we obtain
\begin{align*}
\Phi_{f,\hat{f}} = \frac{1}{2^{3n/2}}\displaystyle \sum_{\var{x} \in \{0,1\}^n}f(\var{x})  W_{\hat{f}}(\var{x})= \frac{1}{2^{3n/2}}\displaystyle \sum_{\var{x} \in \{0,1\}^n} 2^{n/2} \left(f(\var{x})\right)^2 = 1.
\end{align*}
Similarly, for two bent functions $f$ and $g$, having $\Phi_{f,g}=1$ 
necessarily implies that $g$ and $W_f(\bm\omega)$ always agree on the signs. Hence, $g$ can be written as $g(\bm\omega)=\frac{1}{2^{n/2}}W_f(\bm\omega)$ implying $f$ and $g$ are dual to each other.
\end{proof}

Furthermore, given two bent functions $f$ and $g$, they are said to be anti-dual if they satisfy $f(\bm\omega)=-\frac{1}{2^{n/2}}W_g(\bm\omega)$ and $g(\bm\omega)=-\frac{1}{2^{n/2}}W_f(\bm\omega)$. Similarly, we can also define the self-anti-dual bent functions.
In this direction, we have the following simple corollaries.
\begin{corollary}
Given $f,g:\{0,1\}^n\rightarrow \{1,-1\}$, two bent functions. Then, $\Phi_{f,g}=-1$ if and only if $f$ and $g$ are anti-dual to each other.
Similarly, given $f$ is a bent function,
$\Phi_{f,f}=1$ if and only if $f$ is self dual and $\Phi_{f,f}=-1$ if and only if $f$ is anti-self-dual.
\end{corollary}

Let us now look into the other extreme where the Forrelation of two bent functions has very low
absolute value. Together with Proposition~\ref{prop:1} this gives various instantiations of the desired scenario. 
The best scenario is obviously the scenario where $\Phi_{f,g}=0$. 
This can happen in the following scenario.

\begin{proposition}
\label{prop:2}
Let $f$ and $g$ be two bent functions such that $f \op g$ is a balanced function. 
Then the Forrelation $\Phi_{f,\hat{g}}=\Phi_{g,\hat{f}}=0$ where $\hat{f}$ and $\hat{g}$
are duals of $f$ and $g$, respectively.
\end{proposition}

\begin{proof}
From the definition of Forrelation, we obtain
$$\Phi_{f,\hat{g}} = \frac{1}{2^{3n/2}}\displaystyle \sum_{\xv \in \{0,1\}^n} f(\xv) W_{\hat{g}}(\xv)= \frac{1}{2^{3n/2}}\displaystyle \sum_{\xv \in \{0,1\}^n} f(\xv) g(\xv) 2^{\frac{n}{2}}=\frac{1}{2^n}\displaystyle \sum_{\xv \in \{0,1\}^n} f(\xv) g(\xv).$$

Since we know, $f \op g$ is balanced, for $2^{n-1}$ inputs, $f(\xv)=g(\xv)$ and for the other half 
we have $f(\xv)=-g(\xv)$. Let us denote these sets of inputs as $S$ and $\overline{S}$
respectively. 
Then we can write $\Phi_{f,\hat{g}}$ as follows. 
\begin{align*}
\Phi_{f,\hat{g}}= \frac{1}{2^{n}}\left[\displaystyle \sum_{\xv \in S} f(\xv) g(\xv) 
+
\displaystyle \sum_{\xv \in \overline{S}} f(\xv) g(\xv) \right]= \frac{1}{2^{n}}\left[\displaystyle \sum_{\xv \in S} {f(\xv)}^2 
+
\displaystyle \sum_{\xv \in \overline{S}} -{f(\xv)}^2 \right]
=\frac{1}{2^n}\left[2^{n-1}-2^{n-1} \right]=0.
\end{align*}
The derivation of $\Phi_{g,\hat{f}}$ follows along the same lines.
\end{proof}

Combining the statements of Propositions~\ref{prop:1} and \ref{prop:2}, 
we have the following promise problem which can be deterministically determined
using $2$-fold Forrelation.

\begin{proposition}
\label{pth:1}
Given oracle access to a bent function $f$, the problem of finding whether
another given bent function $g$ is $f$'s dual ($\hat{f}$) or if $g \op \hat{f}$
is balanced can be deterministically decided by making exactly one query to $f$ and $g$ each.
\end{proposition}
\begin{proof}
The $2$-query Forrelation algorithm is run with $f_1=f$ and $f_2=g$. 
If $g=\hat{f}$ then $\Phi_{f,g}=1$, resulting the all zero state upon measurement, with certainty.
On the other hand, if $g \op \hat{f}$ is balanced then 
$\Phi_{f,g}=0$ implying that the all zero output is never appears. This allows us 
to deterministically identify the relation between $f$ and $g$.
\end{proof}

Next we look into another scenario where $\Phi_{f,g} \neq 0$ but has very low value. Then the promise problem cannot be solved deterministically, but we can resolve the problem with good probability using only a constant number of queries. In this regard, we consider the 
Kerdock codes and let us denote this set as ${\cal K}$.

Given two bent functions $f_1, f_2 \in {\cal K}$, we know that $f_1 \op f_2$ is bent. Moreover, we know that in the output of a bent 
function, the number of ones and zeros differ by $2^{\frac{n}{2}}$. Thus for $f_1$ and $f_2$ we have 
$\sum_{\xv \in \{0,1\}^n}f_1(\xv)f_2(\xv)=\pm 2^{\frac{n}{2}}.$
Based on this simple observation we construct a promise problem, which is another
desirable instantiation of the $2$-fold Forrelation problem, given in \cite{aaron}.

\begin{proposition}
\label{pth:2}
Given two $n$-variable bent functions $f, g$ such that $f, g, \hat{f} \in {\cal K}$,
the problem of determining whether $g = \hat{f}$ or not can be efficiently solved by making a query each to $f$ and 
$g$ using the $2$-query $2$-fold Forrelation algorithm proposed by \cite{aaron}.
\end{proposition}
\begin{proof}
From Proposition~\ref{prop:1}, if $g$ is dual of $f$, then $\Phi_{f,g}=1$. 
If $g$ is another bent function in ${\cal K}$, then we have 
\begin{align*}
\Phi_{f,g}
&=\frac{1}{2^{\frac{3n}{2}}}
 \displaystyle \sum_{\xv \in \{0,1\}^n} g(\xv)W_f(\xv)=\frac{1}{2^{\frac{3n}{2}}}
\displaystyle \sum_{\xv \in \{0,1\}^n}
g(\xv)\hat{f}(\xv)2^{\frac{n}{2}}=\displaystyle \frac{1}{2^n}\sum_{\xv \in \{0,1\}^n} g(\xv)\hat{f}(\xv) =\pm \frac{1}{2^n}2^{\frac{n}{2}}=\pm \frac{1}{2^{\frac{n}{2}}}.
\end{align*}
For $n>14$ we have $\left|\Phi_{f,g}\right|=\left|\frac{1}{2^{\frac{n}{2}}}\right|<\frac{1}{100}$ and in all such cases
we have the said desirable instantiation following Theorem~\ref{th:aaron}.
\end{proof}

\section{Algorithms related to Walsh Spectrum via $3$-Fold Forrelation}
\label{sec:4}
In this section and in Section~\ref{sec:5} we look into the $3$-fold 
Forrelation set-up and use it as a unifying framework for studying different 
Boolean function spectra, namely the Walsh Spectrum, the cross-correlation spectrum and the 
autocorrelation spectrum. In this section we first study two algorithms 
for the $3$-fold Forrelation, $\Af$ and $\As$. Then we use this formulation for building
Deutsch-Jozsa like algorithms for Walsh spectrum sampling and estimating
if a function is $m$-resilient. We first use the $2$-query algorithm for $3$-fold 
Forrelation to obtain results equivalent to the Deutsch-Jozsa algorithm and then we
use the $3$-query algorithm for the same problem (both the algorithms are due to~\cite{aaron}) 
to outperform the Deutsch-Jozsa algorithm in terms of Walsh spectrum sampling.
The main idea behind the results in this section is intelligently choosing one of the three
functions of the $3$-fold Forrelation formulation ($f_2$ in our case) so that putting the
function $f$ as $f_1$ and $f_3$ allows us to check for Walsh spectrum related properties 
using $\Af$ and $\As$. 
The idea is also carried over for the cross-correlation estimation algorithm in Section~\ref{sec:5}, although for designing a cross-correlation sampling algorithm we need further modifications.

\subsection{A suitable representation of $3$-fold Forrelation and the corresponding algorithms}
\label{sub:representation}
Given three functions $f_1,f_2,f_3:\{0,1\}^n \rightarrow \{1,-1\}$, the $3$-fold Forrelation is defined as 
$$\Phi_{f_1,f_2, f_3}= \frac{1}{2^{2n}}\sum\limits_{\var{x_1}, \var{x_2} , \var{x_3} \in \{0,1\}^n}
f_1(\var{x_1})(-1)^{\var{x_1} \cdot \var{x_2}} f_2(\var{x_2}) (-1)^{\var{x_2} \cdot \var{x_3}} f_3(\var{x_3}).$$
Interestingly, we can also write it down in the following manner.
\begin{align*}
\Phi_{f_1,f_2, f_3}=& 
\frac{1}{2^{2n}}\sum\limits_{\var{x_1}, \var{x_2} , \var{x_3} \in \{0,1\}^n}
f_1(\var{x_1})(-1)^{\var{x_1} \cdot \var{x_2}} f_2(\var{x_2}) (-1)^{\var{x_2} \cdot \var{x_3}} f_3(\var{x_3})
\\=&
\frac{1}{2^{2n}}\sum\limits_{\var{x_2} \in \{0,1\}^n} f_2(\var{x_2})
\sum\limits_{\var{x_1} \in \{0,1\}^n} f_1(\var{x_1})(-1)^{\var{x_2} \cdot \var{x_1}} 
\sum\limits_{\var{x_3} \in \{0,1\}^n} f_3(\var{x_3})(-1)^{\var{x_2} \cdot \var{x_3}}
\\ =&
\frac{1}{2^{2n}}\sum\limits_{\var{x_2} \in \{0,1\}^n} f_2(\var{x_2}) W_{f_1}(\var{x_2}) W_{f_3}(\var{x_2}).
\end{align*}
Moreover, using $f_1=f_3=f$ and $f_2=g$, this gives us 
$\Phi_{f,g, f}=
\frac{1}{2^{2n}}\sum\limits_{\var{x_2} \in \{0,1\}^n} g(\var{x_2}) {W_{f}(\var{x_2})}^2.$
Thus, the $3$-fold-Forrelation actually allows us to sample different combinations of the squares of the Walsh spectrum values of the function, and
when all three functions are different, different combinations of point-wise product
of Walsh spectrum values of two functions, where the combinations are decided by the function 
$f_2$. In all the algorithms we shall design and describe going ahead, we will modify the functions
$f_1,f_2$ and $f_3$ along with some additional tweaks to obtain the desirable sampling scenario.

Given oracle access to $f_1,f_2,f_3:\{0,1\}^n\rightarrow \{-1,1\}$, defined on $n$ variables, let us now briefly discuss the $3$-query and $2$-query quantum algorithms for $3$-fold Forrelation, due to \cite{aaron}).

\begin{enumerate}
\item \textbf{The $3$-query algorithm.} The $3$-query algorithm begins with an $(n+1)$ qubit state $\ket{0}^{\otimes n}\ket{-}$ and effectively
calls to $U_{f_1}$, $U_{f_2}$ $U_{f_3}$ with $H^{\otimes n}$ implemented at the beginning,
between the oracle calls and at the end. Finally, after ignoring the last qubit, the amplitude corresponding to the all zero state is given by $\Phi_{f_1,f_2,f_3}$.
We represent this algorithm as a unitary operation $\As(f_1,f_2,f_3)$ 
which has access to the oracles of the functions $f_1, f_2$ and $f_3$ 
and finally results an $n$-bit output where the probability of 
all the bits being $0$ is $\left( \Phi_{f_1,f_2,f_3} \right)^2$. The quantum structure of the algorithm $\As(f_1,f_2,f_3)$ is depicted in Figure~\ref{fig:q3f3}.
\begin{figure}[H]
\begin{center}
\includegraphics[scale=0.9]{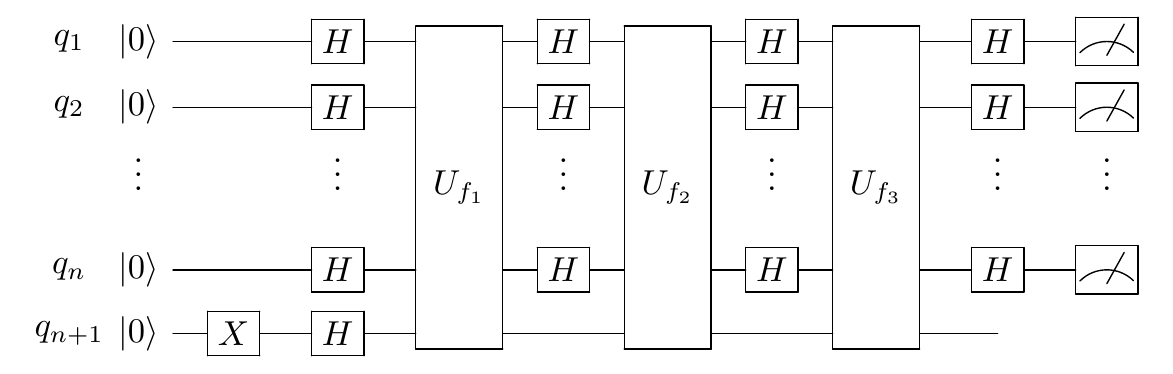}
\caption{Quantum circuit for implementing the $3$-fold Forrelation problem using $3$ queries.}
\label{fig:q3f3}
\end{center}
\end{figure}
\item \textbf{The $2$-query algorithm.} The $2$-query algorithm begins with an $(n+2)$ qubit state $\ket{+}\ket{0}^{\otimes n}\ket{-}$, where the first qubit is the driving qubit. Controlled on the driving qubit being in the $\ket{0}$ state we sequentially apply $H^{\otimes n}\rightarrow U_{f_1}\rightarrow H^{\otimes n}\rightarrow U_{f_2} \rightarrow H^{\otimes n}$ and controlled on the driving qubit being in the $\ket{1}$ state, we fist apply $H^{\otimes n}$ and then apply the oracle $U_{f_3}$, where $H^{\otimes n}$ is applied on the $n$ query-registers (refer to Figure \ref{fig:q2f3}) and the oracles are applied on all but the driving qubit.
Finally, we measure the first qubit in the Hadamard basis, which is equivalent to applying a Hadamard gate in the first qubit followed by the measurement in the $\{\ket{0},\ket{1}\}$ basis. As a result, the probability of obtaining $0$ upon measuring the driving qubit becomes $\frac{1}{2}\left(1+\Phi_{f_1,f_2,f_3}\right)$.
We represent this algorithm as a unitary operation $\Af(f_1,f_2,f_3)$ 
which has access to the oracles of the functions $f_1, f_2$ and $f_3$ 
and finally outputs $0$ with probability $\frac{1+\Phi_{f_1,f_2,f_3}}{2}$
and $1$ with probability $\frac{1-\Phi_{f_1,f_2,f_3}}{2}$. The quantum structure of this algorithm $\Af(f_1,f_2,f_3)$ is depicted in Figure~\ref{fig:q2f3}.
\begin{figure}[H]
\begin{center}
\includegraphics[scale=0.45]{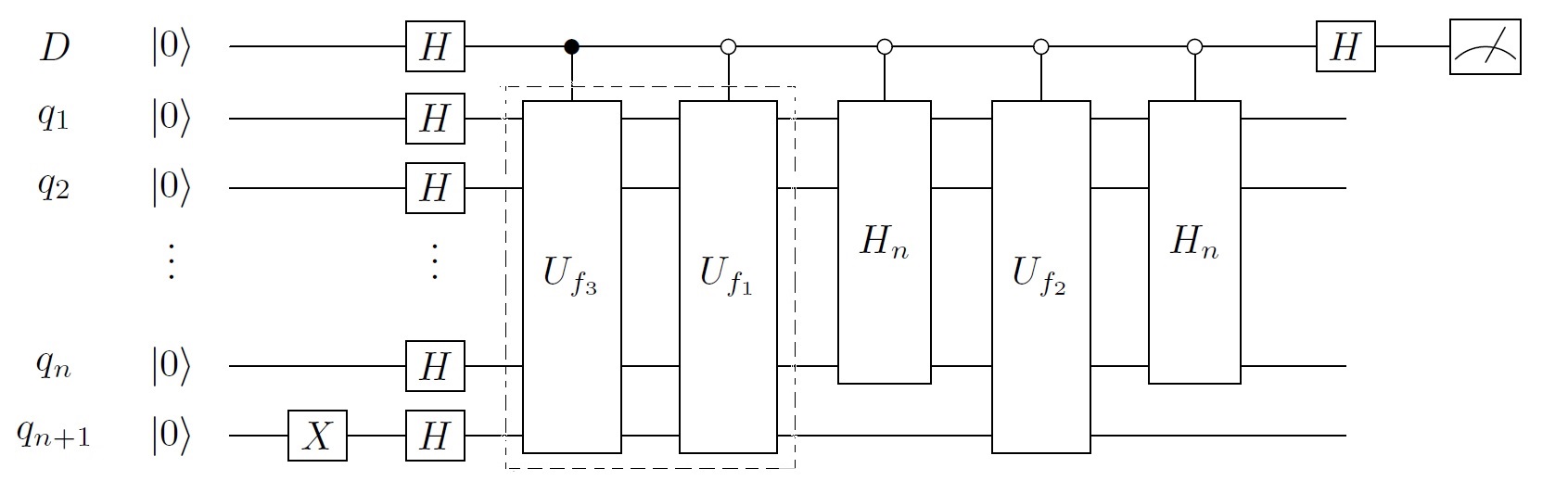}
\caption{Quantum circuit for implementing the $3$-fold Forrelation problem using $2$ queries.}
\label{fig:q2f3}
\end{center}
\end{figure}
\end{enumerate}
We now look into the Walsh spectrum sampling of a Boolean function $f$ at any given point $\bm{\omega}$ using the algorithms $\Af$ and $\As$ for the $3$-fold Forrelation and compare the sampling probabilities with that of the Deutsch-Jozsa algorithm.

\subsection{Walsh Spectrum Sampling using $3$-fold Forrelation}
\label{sub:sampling}
Consider an $n$-variable Boolean function $f$ and a set of points $S \subseteq\{0,1\}^n$. Then the probability of obtaining any one of the states from $S$ in the outcome of the Deutsch-Jozsa algorithm for $f$ is given by 
\begin{equation}
\label{eqss1}
p=\displaystyle\frac{1}{2^{2n}}\sum_{\xv\in S}{W_f(\xv)^2}.
\end{equation}
Next we sample the Walsh spectrum values of $f$ on the set $S$ using $3$-fold Forrelation. We define the Forrelation set up for $\Phi_{f_1,f_2,f_3}$ as follows. We have $f_1=f_3=f$ and the function $f_2=g$ is designed by us, depending upon the set $S$ such that $g(\xv)=-1$ if and only if $\xv\in S$ and $g(\xv)=1$, otherwise. Then we have,
$$\Phi_{f,g,f} =\frac{1}{2^{2n}}\sum_{\xv\in\{0,1\}^n}g(\xv)W_f(\xv)^2 = \frac{1}{2^{2n}} \left(\sum_{\xv \not\in S}W_f(\xv)^2 - \sum_{\xv\in S}{W_f(\xv)^2}\right).$$
Now applying the Parseval's identity, $\sum_{\xv \in \{0,1\}^n} W_f(\xv)^2=\sum_{\xv \not\in S}W_f(\xv)^2 + \sum_{\xv\in S}{W_f(\xv)^2}= 2^{2n}$ we obtain,
$$\Phi_{f,g,f} = \frac{1}{2^{2n}} \left( 2^{2n}-\sum_{\xv\in S}{W_f(\xv)^2} - \sum_{\xv\in S}{W_f(\xv)^2} \right)
= 1-\frac{2}{2^{2n}}\sum_{\xv\in S}{W_f(\xv)^2}.$$
In this regard, we have the following result.
\begin{lemma}
\label{lem:2q}
The probability of getting the output $1$ upon running the algorithm $\Af(f,g,f)$ for $3$-fold Forrelation where $g(\xv)=-1,\,\forall\xv\in S$ and $g(\xv)=1$, otherwise, is same as $p$ as in Equation~\eqref{eqss1}.
\end{lemma}
\begin{proof}
From $\Af(f,g,f)$, the probability of getting the output $1$ is
$$\frac{1-\Phi_{f,g,f}}{2}= \frac{1}{2}\left[1-\left(1-\frac{2}{2^{2n}}\sum_{\xv\in S}{W_f(\xv)^2}\right)\right]
=\frac{1}{2^{2n}}\sum_{\xv\in S}{W_f(\xv)^2}=p.$$
\end{proof}
Thus the $2$-query algorithm $\Af(f,g,f)$ makes a single query to $f$ and another query to $g$, designed based on the set $S$, behaves equivalent to the Deutsch-Jozsa algorithm in terms of Walsh spectrum sampling. Now we show that $\Af(f,g,f)$ can be used here to obtain the improvement.

\begin{theorem}
\label{thm:3q}
The probability of getting an output with at least one bit being $1$ upon running the algorithm $\As(f,g,f)$ for $3$-fold Forrelation is given by $4p-4p^2$, where $p$ is as in Equation~\eqref{eqss1}.
\end{theorem}
\begin{proof}
From Lemma \ref{lem:2q}, we obtain $\frac{1-\Phi_{f,g,f}}{2}=p$, which implies $\Phi_{f,g,f}=1-2p$. Now, the probability of getting an output with at least one bit being $1$ upon running the algorithm $\As(f,g,f)$ is 
$1-\Phi_{f,g,f}^2= 1-\left(1-2p \right)^2=4-4p^2.$
\end{proof}

For $p < 0.75$, we have $4p-4p^2 > p$. Moreover, if one argues that 
$\As$ makes two queries to $f$ compared one by Deutsch-Jozsa, 
it is easy to see that the probability of observing any one of the states from $S$ at least once by running the Deutsch-Jozsa algorithm twice is $1- (1-p)^2=2p-p^2$
which is also lower than the probability due to $\As$ for small values of $p$.
In fact we can compare the results due to $\As$, sampling from Deutsch-Jozsa once, sampling from Deutsch-Jozsa twice and sampling using Deutsch-Jozsa followed by one round amplitude amplification
for small values of $p$ in the following manner. 

\begin{enumerate}
\item
The probability of observing any one of the states from $S$ after running the DJ algorithm is $p$. This requires one query made to the oracle of $f$.
\item
The probability of observing any one of the states from $S$ at least once, after running Deutsch-Jozsa algorithm
twice is $2p-p^2 \approx 2p$.
This requires two queries to the oracle of $f$.
\item
If we first apply the Deutsch-Jozsa and then amplify the states from $S$ using a single round of amplitude amplification,
then at first we have $\theta=\sin^{-1}p$ and after a round the probability becomes $\sin \left( 3\sin^{-1}p \right)$.
Since for small values of $\theta$ we have $\theta \approx \sin \theta$, the probability of observing any one of the states from $S$ becomes $\approx 3p$. This method also requires two queries to $f$, one for the initial 
Deutsch-Jozsa and once for designing the inversion about mean operator.

\item
Finally, if we use the algorithm $\As(f,g,f)$, where $g(\xv)=-1,\,\forall\xv\in S$ and $g(\xv)=1$ otherwise,
then the probability obtaining atleast one bit being $1$, is given by $4p-4p^2\approx 4p$ which outperforms all the techniques discussed above for small values of $p$.
\end{enumerate}
Figure~\ref{fig:graph} presents a schematic view of these probability values.

\begin{figure}[H]
\begin{center}
\includegraphics[scale=0.48]{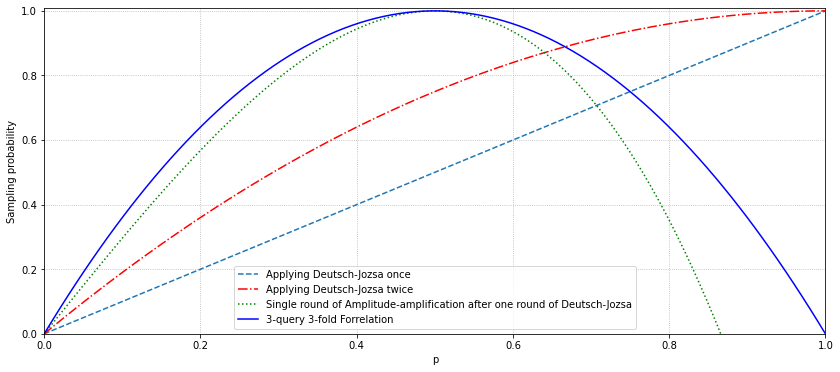}
\caption{Sampling probabilities of Walsh transform using different algorithms}
\label{fig:graph}
\end{center}
\end{figure}

Hence, the $3$-query algorithm for $3$-fold Forrelation, $\As(f,g,f)$ samples the Walsh spectrum values of $f$ at any given set of points, $S$ more efficiently compared to the Deutsch-Jozsa algorithm.
This result underlines how Forrelation can be used for efficient sampling of the Walsh spectra of a Boolean functions at any given point.

Observe that for $\left|S\right|=1$, the $3$-fold Forrelation algorithm $\As(f,g,f)$ samples the Walsh transform value of $f$ at any given point more efficiently compared to the DJ algorithm.
Moreover, if we consider $S$ such that $S=\{\xv:wt(\xv)\leq m\}$, then using $\As(f,g,f)$ we obtain a better sampling of Walsh spectra for the points with Hamming weight less than or equal to a particular weight $m$ compared to the DJ algorithm. This provides a constant improvement over the state-of-the-art result \cite{chak}.

\subsection{Implications to Resiliency checking}
Given the problem where one has to decide whether a function $f$ is $m$-resilient, we can simply apply Deutsch-Jozsa on $f$ and obtain the resultant state, $\frac{1}{2^n}\sum_{\xv \in \{0,1\}^n} W_f(\xv)\ket{\xv}$.
If the function is indeed $m$-resilient the probability of getting an output of weight 
less than or equal to $m$ is zero. However, in the other case, the probability of observing such a state $\xv$ with $wt(\xv)\leq m$ depends upon how much of the Walsh spectrum mass is distributed on points 
with weight $\leq m$.
Therefore, absence of such a state $\xv$ with $wt(\xv)\leq m$ can not conclude that $f$ is `$m$-resilient'. However, if we observe one such state even once, that is 
sufficient to conclude that the function $f$ is not $m$-resilient.
The states $\xv$ with $wt(\xv)\leq m$ 
can be termed as an `undesirable outcome', presence of which deterministically concludes about $f$ being non-resilient. Thus any algorithm should try to 
determine whether $f$ is ``not $m$-resilient" to the best of its ability and declare $f$ 
is $m$-resilient if it does not obtain any undesirable outcome after certain number 
of executions of the algorithm. This was also the methodology of~\cite{chak}.

Here we have the following observations.
\begin{itemize}
\item The function $4p-4p^2$ becomes a decreasing function 
for $p>\frac{1}{2}$. This is however very simple to deal with. 
At the beginning, we sample with Deutsch-Jozsa some constant number of times and see if 
we obtain any undesirable outcome. If $p$ is indeed greater than $\frac{1}{2}$ 
then we will obtain an undesirable outcome with a high probability. Otherwise, 
we apply $\As(f,g,f)$ which samples lower values of Walsh spectra better compared to DJ algorithm.
\item If $p$ is very low, the probability of obtaining an undesirable outcome also becomes very low in both the Deutsch-Jozsa set-up and the $3$-fold Forrelation set-up. 
Against this backdrop the amplitude amplification algorithm~\cite{brass} 
was implemented by~\cite{chak} on top of the Deutsch-Jozsa 
algorithm to amplify the amplitude of the undesirable state and increase the
probability of obtaining an undesirable state, if one exists for the given function.
The outcome of our algorithm due to $\As$ can also be amplified 
in a similar fashion, and the overall circuit complexity is of the same asymptotic order.
\end{itemize}

Since, the starting state of $\As$ makes the sampling probability of an undesirable outcome
$4$-times more likely compared to Deutsch-Jozsa, therefore, we need to run the amplification iterate one-fourth of the times to obtain the similar result compared to the result in~\cite{chak}. 
However the application of $\As$ requires two oracle
queries to $f$. 
Therefore if the amplification iterate is run $n$ times for $\As$ then it requires
a total of $2n$ many queries to $f$, compared to $4n$ queries in case of \cite{chak}. 
Therefore in total the query complexity due to $\As$ is half as compared to the 
result by~\cite{chak}.
The focus of this paper is to design the starting states before the sampling algorithms 
and therefore further study about the detailed procedure is beyond the scope of this paper.

Finally, in the next section, we discuss how this algorithm can also be used to sample cross-correlation values of two 
functions by simply replacing the symmetric function that is $f_2$'s placeholder, with a linear function of our choice. We also 
consider tweaking the quantum algorithm for 3-fold Forrelation to obtain further results.

\section{Cross-correlation and $3$-fold Forrelation}
\label{sec:5}
The cross-correlation $C_{f,g}$ of two functions $f$ and $g$ relates to the 
property of Shannon's confusion in a cryptographic system, and is an widely
studied cryptographic property. Cross-correlation, and autocorrelation 
spectra have a very significant difference from the Walsh Spectrum. 
For the Walsh Spectrum of any function $f$, we have 
$\frac{1}{2^{2n}}\sum_{\xv \in \{0,1\}^n}W_f(\xv)^2=1$.
This lies in the heart of the Deutsch-Jozsa algorithm, which allows us to 
obtain any $n$ bit string $\xv$ with probability $\frac{W_f(\xv)^2}{2^{2n}}$
with the sum of probabilities adding upto $1$. There is no such identity 
for the cross-correlation spectrum. We have the bound of 
$2^{2n} \leq \sum_{\xv \in \{0,1\}^n} C_{f,g}(\xv)^2 \leq 2^{3n}$.
This stops us from designing a Deutsch-Jozsa like sampling algorithm 
for cross-correlation and autocorrelation spectra. 

In the work~\cite{bera} Bera et al obtained an autocorrelation 
spectrum sampling algorithm that produced a $2n$ bit state $\xv||0^n$
with probability $\frac{C_f(\xv)^2}{2^{3n}}$ and also autocorrelation 
estimation algorithms in the quantum framework. 
The former result was based on studying the connection between Walsh spectra
of derivatives of $f$ and the autocorrelation spectrum of $f$, and thus the
same algorithm cannot be extended for cross-correlation in its current form. 
Against this backdrop, we design cross-correlation estimation and sampling algorithms
(and as corollaries they also imply to the autocorrelation spectrum)
in this section, starting with cross-correlation estimation. 

We start with a very interesting observation of~\cite{cross} which leads us to our 
first algorithm.

\begin{theorem}[\cite{cross}]
\label{th:cross}
Given any two functions $f$ and $g$, we have
$$[C_{f,g}(000\ldots 0), \ldots ,C_{f,g}(111\ldots 1)] \hat{H}_n= [W_fW_g(000\ldots 0), \ldots ,W_fW_g(111\ldots 1)]$$
where $\hat{H}_n= 
{\begin{pmatrix}
1 &1\\
1 &-1
\end{pmatrix}}^{\otimes n}
.$
\end{theorem}

Here, observe that $\hat{H}_n$ is a $2^n \times 2^n$ ``non-normalized'' Hadamard matrix and thus 
we have $\left( \hat{H}_n \right)^2= 2^n I_n$  where $I_n$ is the $2^n$-dimensional identity matrix. 
Thus multiplying $\hat{H}_n$ to both sides of the above identity results in the following equation.
\begin{equation}
\label{eq:1}
[C_{f,g}(000\ldots 0), \ldots ,C_{f,g}(111\ldots 1)] 2^n I_n=
[W_fW_g(000\ldots 0), \ldots ,W_fW_g(111\ldots 1)] \hat{H}_n.
\end{equation}

We also fix a general notation to denote the $2^n$ linear functions on $n$ variables,
denoting the functions as $\mathbb{L}_{\yv}(\xv) =(-1)^{\oplus_{i:y_i=1}x_i}$. 
In this regard we have the following well known result.

\begin{fact}
\label{fact:1}
The $2^n$ linear functions have a one-to-one correspondence with the columns of $H_n$.
That is, $H_n[\var{i}][\var{j}]=
\mathbb{L}_{\var{i}}(\var{j}), ~\var{j} \in \{0,1\}^n$ 
for all $\var{i} \in \{0,1\}^n$.
\end{fact}

Using these results along with the $3$-fold Forrelation formulation we get the following results.

\begin{theorem}
\label{th:main}
Given oracle access to two functions $f$ and $g$ on $n$ variables,
two algorithms $A_1$ and $A_2$ can be designed in a way that they 
make one query to $f$ and $g$ each and makes another 
query to a linear function $\mathbb{L}_{\yv}$ such that 

\begin{enumerate}

\item 
In algorithm $A_1$, upon measuring $n$ predetermined qubits in the computational basis,
the probability of obtaining the all zero state is $\frac{\left( C_{f,g}(\yv) \right)^2}{2^{2n}}$.

\item 
In algorithm $A_2$, upon measuring $1$ predetermined qubit in the computational basis, the probability of obtaining
the output $0$ is  $\frac{1}{2}\left(1+\frac{C_{f,g}(\yv)}{2^n}\right)$.

\end{enumerate}

\end{theorem}

\begin{proof}
Given the functions $f,g$ and the linear function $\mathbb{L}_{\yv}$,
let $A_1$ be the algorithm $\As(f,\mathbb{L}_{\yv},g)$. 
Then the probability of obtaining the all zero state upon measurement is given by $\left( \Phi_{f,\mathbb{L}_{\yv},g} \right)^2$
where the Forrelation value equates to
\begin{align*}
\Phi_{f,\mathbb{L}_{\yv},g}
&=\frac{1}{2^{2n}}
\sum\limits_{\xv \in \{0,1\}^n} W_f(\xv)W_g(\xv)\mathbb{L}_{\yv}(\xv)
=\frac{1}{2^{2n}}
\sum\limits_{\xv \in \{0,1\}^n} W_f(\xv)W_g(\xv)H_n[\yv][\xv]=\frac{2^nC_{f,g}(\yv)}{2^{2n}}=\frac{C_{f,g}(\yv)}{2^{n}}.
\end{align*} 
Thus the probability of getting the all zero state as measurement outcome is
$\frac{\left( C_{f,g}(\yv) \right)^2}{2^{2n}}$.

Similarly we can define $A_2=\Af(f,\mathbb{L}_{\yv},g)$ and thus we obtain 
the output $0$ with probability $\frac{1}{2}\left(1+\frac{C_{f,g}(\yv)}{2^n}\right)$.
\end{proof}

This gives us a constant query algorithm for sampling the cross-correlation value 
of any two functions at any given point. Here note that one needs to design different
algorithms (and thus circuits) in order to obtain the cross-correlation value at different 
points. Let us first compare this results with the results by~\cite{bera} on the autocorrelation which is $C_{f,f}$ and is simply denoted as $C_f$. 

\subsection{Comparison of auto-correlation results with~\cite{bera}}
\label{compbera}
In the first algorithm for autocorrelation-sampling described in Theorem~\ref{bera:auto}, the probability of obtaining the $2n$ bit 
string $0^n||\yv$ is $\frac{C_f(\yv)^2}{2^{3n}}$. In comparison the 
algorithm $A_1$ designed by us in Theorem~\ref{th:main} outputs the $n$ bit state $0^n$ with probability  $\frac{C_f(\yv)^2}{2^{2n}}$.
Thus $A_1$ provides the autocorrelation value at a particular point more
efficiently, although it does not sample from the autocorrelation spectrum. 
Here the following remark is important. 

\begin{remark}
Note that the algorithm described in Theorem~\ref{bera:auto} from~\cite{bera} is designed specifically for sampling 
autocorrelation spectrum of functions which relies on obtaining Walsh spectrum of derivatives of the 
function $f$. In this direction we obtain an algorithm to sample from the 
cross-correlation spectrum of any two functions $f$ and $g$ which we obtain 
by tweaking the $3$-query $3$-fold Forrelation algorithm $\As$.
\end{remark}

Next we discuss amplitude-estimation. The second result by~\cite{bera} showed
how one can have an estimation algorithm such that 
given a function $f$ there exists an algorithm that makes 
$\mathcal{O}\left( \frac{\pi}{\epsilon} \log \frac{1}{\delta} \right)$ query
and returns an estimate $\alpha$ with
$\pr\big[\alpha-\epsilon \leq \frac{C_f(\yv)^2}{2^{2n}} \leq \alpha+\epsilon \big] \geq 1-\delta$.
Here the algorithm estimates $\frac{\left(C_f{(\yv)}\right)^2}{2^{2n}}$ as probability of outcome 
of a particular state in the algorithm described in Lemma~\ref{bera:auto-est} 
is $\frac{1}{2}\left(1+\frac{C_f{(\yv)}^2}{2^{2n}}\right)$ on which the result of Lemma~\ref{bera:est}
is applied. In this regard we have the following result. 

\begin{theorem}
\label{th:usest}
Given a function $f$, there exists an algorithm that makes 
$\mathcal{O}\left( \frac{\pi}{\epsilon} \log \frac{1}{\delta} \right)$ queries to the oracle of $f$ and returns an estimate $\alpha$ such that
$$\pr[\alpha-\epsilon \leq \frac{C_{f,g}(\yv)}{2^n} \leq \alpha+\epsilon ] \geq 1-\delta.$$
\end{theorem}

\begin{proof}
We have the algorithm $A_2$ that outputs the state $0$ with probability 
$\hat{P}(f,g)=\frac{1}{2}\left(1+\frac{C_{f,g}(\yv)}{2^n}\right)$. 
This allows us to estimate $\hat{P}(f,g)$ using the result of Lemma~\ref{bera:est}. 
It is easy to see that estimating 
$\frac{1}{2}\left(1+\frac{C_{f,g}(\yv)}{2^n}\right)$ and $\frac{C_{f,g}(\yv)}{2^n}$ are equivalent,
which gives us the requisite result.
If we set $f$ and $g$ to be the same function, this translates into an algorithm 
for autocorrelation estimation. 
\end{proof}

\begin{remark}
Here note that the query complexity needed to $\epsilon$-estimate $\frac{C_f(\yv)}{2^n}$
in Theorem~\ref{th:usest} is the same as the query complexity needed to $\epsilon$-estimate
$\left(\frac{C_f(\yv)}{2^n} \right)^2$ in the work~\cite{bera}. 
This implies for comparative accuracy of $\frac{C_f(\yv)}{2^n}$, we need square-root 
of the number of queries, which is again an improvement on~\cite{bera}.
\end{remark}

Finally we move towards a cross-correlation sampling algorithm, which is the final 
contribution of our paper. As a corollary of this result we can also check 
if two functions are uncorrelated of degree $m$ for all values of $m$ upto a bound 
so that the time complexity is better than that of general classical algorithms. 
For the later part, we use the fact of existence of polynomial sized circuits 
for Dicke state preparation. Let us now define a Dicke state, denoted by $\ket{D^n_{k}}$.
\begin{definition}
\label{def:dicke}
An $n$-qubit quantum state with equal superposition of all $\binom{n}{k}$ many basis states of weight $k$ is called a Dicke state.
\end{definition}

\subsection{A cross-Correlation sampling algorithm}
Now we move towards cross-correlation sampling.
In case of sampling the cross-correlation value of two functions at a point,
we used the $3$-fold Forrelation set-up, where the second function is a linear
function of our choice. 

As a first step, we modify $\As$ so that the second function 
$\mathbb{L}_{\yv}$ can be decided by additional inputs.
Let us now look into the quantum circuit $\As$ 
corresponding to this algorithm. The corresponding circuit has $n+1$ qubits, 
which we denote in the following way.
\begin{enumerate}

\item The $n$ qubits that are used to query the oracles at different 
inputs are named from $q_1$ to $q_n$.

\item There is another qubit that stays in the $\ket{-}$ state throughout 
the circuit and used for phase-kickback, we call this $q_{n+1}$.
\end{enumerate}

Next we describe a very simple implementation of any linear function in the
quantum black-box model. 
Consider the operation $\sf CNOT^S_b$ which denotes the multi-controlled 
NOT operation where $S$ is a set consists of the control-qubits and the qubit $b$ is the target.
Then it is easy to see that any linear function $\mathbb{L}_{\yv}$ can be implemented as
the series of operations $CNOT^{q_i}_b$ such that $y_i=1$. Using this construction 
we tweak the algorithm.

\RestyleAlgo{boxruled}

\begin{algorithm}[ht]
\DontPrintSemicolon
  
  \KwInput{
  \begin{enumerate}
 \item  $\var{u} \in \{0,1\}^n$, 
 \item  $R \otimes Q$   where 
 \subitem $R$ is an $n$-qubit register, denoted as $\otimes_{i=1}^n \ket{r_i}$.
 \subitem $Q$ is an $n+1$-qubit register, denoted as $\otimes_{i=1}^{n+1} \ket{r_i}$.

 \item Description of an algorithm $C_n$.\\

  \end{enumerate}  
  }
  \KwOutput{The output $\var{u}||0^n$ with probability $\frac{C_{f,g}(\var{u})}{2^{2n}}$}

{\bf 	Apply the following operations in the given sequence:}

\begin{enumerate}

	\item
    $C_n$ on the register $R$ (In this case defined as $C_n^1(\var{u})$):\\
   	 \quad \quad 	apply $X$ gate on $r_i$ if $u_i=1$. \nonumber \\ 
   
   	\item
	$HX$ gate on $q_{n+1}$.

   	\item 
    $H^{\otimes n}$ on $q_1$ to $q_n$.

   	\item 		   
    Oracle of $f$ on $Q$.

   	\item    
    $H^{\otimes n}$ on $q_1$ to $q_n$.

   	\item  $DC$ defined as:   			
   	$\sf CNOT^{r_i,q_i}_{q_{n+1}}, 1 \leq i \leq n$.

   	\item 
    $H^{\otimes n}$ on $q_1$ to $q_n$.   
		  
   	\item 		   
    Oracle of $g$ on $Q$.
   
   	\item       
    $H^{\otimes n}$ on $q_1$ to $q_n$.  
 
\end{enumerate}  
\caption{The Algorithm  $\mathbb{A}(C_n)$ where $C_n=C^1_n$.}
\label{algo}
\end{algorithm}
Figure~\ref{fig:algo} provides a schematic diagram of the algorithm $\mathbb{A}(C_n)$.\\
\begin{figure}[ht]
\begin{center}
\includegraphics[scale=0.5]{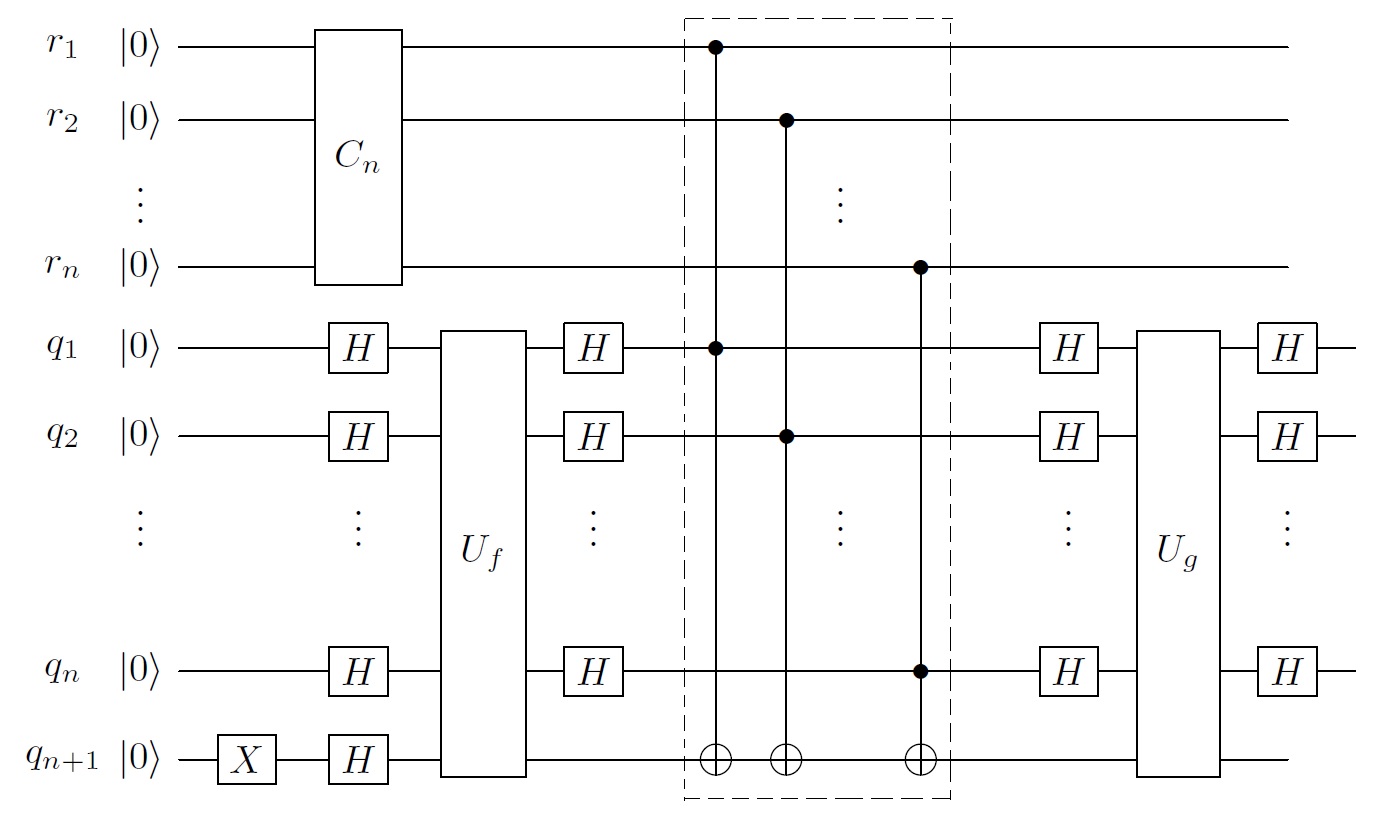}
\caption{Quantum circuit for implementing Algorithm \ref{algo}.}
\label{fig:algo}
\end{center}
\end{figure}

Let us now analyze the working of this algorithm. 
First let us observe a simple fact regarding the operator $DC$.
If $r_i=1$ for some $i$, then the operation 
$CNOT^{r_i,q_i}_{q_{n+1}}$ works as $CNOT^{q_i}_{q_{n+1}}$
which is equivalent to the oracle of the function $h(\xv)=x_i$.
Otherwise if $r_i=0$ then the operation $CNOT^{r_i,q_i}_{q_{n+1}}$
works as identity. Which can be formalized as follows. 

\begin{proposition}
\label{prop:DC}
The operation $DC$ on $R \otimes Q$ can be described as 
$DC\ket{\var{u}}\ket{\xv}\ket{-}= (-1)^{\var{u} \cdot \xv} \ket{\var{u}}\ket{\xv}\ket{-}$.
In essence, the operator $DC$ works as an oracle access to the function $\mathbb{L}_{\var{u}}$
on the register $Q$ when $R$ is in the state $\ket{\var{u}}$.
\end{proposition}
Thus when $C_n$ is defined as $C^1_n(\var{u})$, the algorithm effectively 
becomes same as $\As(f,\mathbb{L}_{\var{u}},g)$ on register $Q$ and the register
$R$ stays in the state $\ket{\var{u}}$ and thus outputs $\var{u}||0^n$
with probability $\frac{C_{f,g}(\var{u})^2}{2^{2n}}$.

The cross-correlation sampling and checking if two functions are uncorrelated of
degree $m$ can be simply obtained by choosing suitable functions $C_n$, which 
is the final contribution of the paper. Let us derive the output state for 
a general $C_n$ operation.

\begin{lemma}
\label{th:tweak}
Let $C_n$ be a $2^n \times 2^n$ unitary operator 
so that $C_n\ket{0}^{\otimes n}= \sum_{\xv \in \{0,1\}^n} \alpha_{\xv}\ket{\xv}$.
Then starting from the state $\ket{0}^{\otimes n} \otimes \ket{0}^{\otimes n+1}$
the algorithm $\mathbb{A}(C_n)$ of Algorithm~\ref{algo} has the pre-measurement state 
$$
\displaystyle \sum_{\var{u} \in \{0,1\}^n}
\alpha_{\var{u}}\ket{\var{u}} 
\left( \frac{C_{f,g}(\var{u})}{2^n}\ket{0^n} + \beta_{\var{u}}\ket{W_{\var{u}}}
\right)
$$
where $\ket{W_\var{u}}$ is an $n$-qubit superposition state such that the amplitude 
of the state $\ket{0^n}$ is $0$.
\end{lemma}
The proof of this lemma follows directly form Algorithm \ref{algo}.
Let us now present the following theorem for cross-correlation sampling using Lemma \ref{th:tweak}. 

\begin{theorem}
\label{th:cor-samp}
If we fix $C_n=H^{\otimes n}$ then upon measuring the predetermined $2n$ many qubits at the end of
$\mathbb{A}(C_n)$, the probability of getting the state $\var{u}||0^n$ is 
$\frac{C_{f,g}(\var{u})^2}{2^{3n}}$ for all $\var{u} \in \{0,1\}^n$. 
\end{theorem}
\begin{proof}
If we fix $C_n=H^{\otimes n}$, then from Lemma~\ref{th:tweak}, the pre-measurement state becomes 
$$
\displaystyle \sum_{\var{u} \in \{0,1\}^n}
\frac{1}{2^{\frac{n}{2}}}\ket{\var{u}} 
\left( \frac{C_{f,g}(\var{u})}{2^n}\ket{0^n} + \beta_{\var{u}}\ket{W_{\var{u}}}
\right).
$$
Thus the probability of getting the output $\var{u}||0^n$ upon measuring the register 
$R$ and qubits $q_1$ to $q_n$ in computational basis is given by
$\frac{C_{f,g}(\var{u})^2}{2^{3n}}$.
\end{proof}
Here note that $\displaystyle{2^{2n} \leq \sum_{\var{u}\in\{0,1\}^n}{C_{f,g}(\var{u})^2} \leq 2^{3n}}$ which implies
$\displaystyle{\frac{1}{2^{3n}}\sum_{\var{u}\in\{0,1\}^n}{C_{f,g}(\var{u})^2} \leq 1}$ and this bound is tight for certain choices of $f$ and $g$. 
Therefore, it is not possible to have a generalized algorithm that always samples 
from the cross-correlation spectrum with better than $\frac{C_{f,g}(\var{u})^2}{2^{3n}}$ probability.

\subsection{Checking if $f$ and $g$ are Uncorrelated of degree $m$}
In this problem we want to check if $C_{f,g}(\xv)=0$ for all $\xv \in \{0,1\}^n$ 
with weight of $\xv$ less than or equal to $m$. This problem is similar to checking
if a function $f$ is $m$-resilient. Thus the same sampling and amplitude amplification
algorithm of~\cite{chak} can be implemented. If there is a quantum algorithm such that the probability of obtaining a good (undesirable) state 
while calling an algorithm is $a^2$, 
then such a state can be obtained with constant probability 
by running the algorithm $\mathcal{O}\left(\frac{1}{a}\right)$ times without the measurement. 
Here if we run the cross-correlation sampling algorithm and if the output is 
of the form $\yv||0^n$ where $wt(\yv) \leq m$ then it concludes $f$ and $g$ are 
not uncorrelated of degree $m$. Then we have the following result.

\begin{proposition}
\label{prop:uncor}
There is an algorithm that verifies if two functions $f$ and $g$ are uncorrelated 
of degree $m$ with constant probability by making $\mathcal{O}(\frac{1}{a})$ queries to $f$ and $g$ each, 
where 
$$a^2= \frac{1}{2^{3n}} \sum_{\xv: wt(\xv) \leq m} C_{f,g}(\xv)^2.$$
\end{proposition}

\begin{proof}
We use the sampling algorithm defined in Theorem~\ref{th:cor-samp}.
Then let an undesirable state be one that guarantees existence of $\var{u}$ 
such that $wt(\var{u}) \leq m$ and $C_{f,g}(\var{u}) \neq 0$. This is same 
as obtaining a state $\var{u}||0^n$ output after measurement of the algorithm 
in Theorem~\ref{th:cor-samp}. The probability of getting such a state is 
$\frac{1}{2^{3n}} \sum_{\xv: wt(\xv) \leq m} C_{f,g}(\xv)^2$. 
Applying the amplitude amplification argument to this yields the result.
\end{proof}

Here note that because of the structure of $\mathbb{A}(C_n)$ we can control the 
points over which cross-correlation is sampled. If $C_n$ is such that the amplitude
of some state $\ket{\yv}$ in the register $R$ is $0$ then $C_{f,g}(\yv)$ never 
affects the output of $\mathbb{A}(C_n)$. Now we can use the fact that the values
of $C_{f,g}(\xv)$ where $wt(\xv)>m$ is irrelevant while trying to check if $f$ and $g$
are uncorrelated of degree $m$. Thus we need a $C_n$ which makes the algorithm efficient 
by not having states whose outputs are depending on $C_{f,g}(\xv),~wt(\xv)>m$. 
To this end we use the following result for Dicke state preparation circuits.
Dicke states $\ket{D^n_i}$ are the equal superposition state of weight $i$ on 
some $n$ qubit system.

\begin{theorem}[\cite{dicke}]
Starting from the state $\ket{0}^n$ any Dicke state $\ket{D^n_m}$ can be 
deterministically prepared using $\mathcal{O}\left(n^2\right)$ CNOT gates and $\mathcal{O}\left(n^2\right)$ many single qubit gates. 
\end{theorem} 

We denote the corresponding unitary that prepares the Dicke state $\ket{D^n_m}$ as $UD^n_m$, 
so that we have 
$$UD^n_m \ket{0}^{\otimes n}= \frac{1}{\sqrt{n \choose m}} \sum\limits_{\xv: wt(\xv)=m} \ket{\xv}.$$
Using this result we design a sampling algorithm in the following manner.

\begin{theorem}
\label{th:dicke}
For any $m<n$, there is an algorithm that verifies if two functions $f$ and $g$ are uncorrelated 
of degree $m$ with constant probability by making $\mathcal{O}(\sum_{i=0}^m\frac{1}{a_i})$ queries to $f$ and $g$ each, 
where 
$$(a_i)^2= \frac{1}{{n \choose i}2^{2n}} \sum_{\xv: wt(\xv) = i} C_{f,g}(\xv)^2$$.
\end{theorem}
\begin{proof}
We proceed in the same way as Proposition~\ref{prop:uncor} but with a different choice of $C_n$. 
Consider the circuit $\mathbb{A}(UD^n_i)$ for some $0 \leq i \leq m$. 
This algorithm by definition of Algorithm~\ref{algo} outputs a state 
$\yv||0^n$ with probability $\frac{C_{f,g}(\yv)^2}{{n \choose i}2^{2n}}$
for all values of $\yv$ with weight $i$, and probability of obtaining the state 
$\xv||0^n$ is zero if $wt(\xv) \neq i$. Thus any output $\yv||0^n$ is an undesirable
outcome corresponding to a point of weight $i$ in the cross-correlation spectrum 
with nonzero value. We can obtain such an outcome by making 
$\mathcal{O}(\frac{1}{a_i})$ queries to $f$ and $g$ each, 
where $(a_i)^2= \frac{1}{{n \choose i}2^{2n}} \sum_{\xv: wt(\xv) = i} C_{f,g}(\xv)^2$.

Thus, to check if there is any such point for all $1 \leq i \leq m$ 
we need to run the circuits $\mathbb{A}(UD^n_i),0 \leq i \leq m$ and apply
amplitude amplification on them individually, and thus the total query complexity
is $\mathcal{O}(\sum_{i=0}^m\frac{1}{a_i})$.
\end{proof}

This is an improvement on the cross-correlation sampling based result for small values of $m$.
If $m$ is a constant then ${n \choose m} < n^m$. Let $M_i$ denote the sum 
$\sum_{\xv: wt(\xv)=i} C_{f,g}(\xv)^2$. 
Then the query complexity for checking if $f$ and $g$ are uncorrelated of degree $m$ using 
the algorithm in Theorem~\ref{th:dicke} would be
$2^{n}\sum_{i=0}^m\sqrt{\frac{{n \choose i}}{M_i}} \leq 2^{n} n^{\frac{m+1}{2}} \sum_{i=0}^m \frac{1}{\sqrt{M_i}}$. 
On the other hand if we would have used Proposition~\ref{prop:uncor} the query complexity would 
have been $2^{\frac{3n}{2}} \sum_{i=0}^m \frac{1}{\sqrt{M_i}}$. For all cases 
where $\sum_{i=0}^m \frac{1}{\sqrt{M_i}}$ is of the form $\frac{1}{poly(n)}$, Theorem~\ref{th:dicke}
provides polynomial improvement over using the cross-correlation sampling algorithm. 

At this point one may wonder if a similar Dicke state oriented approach could be applied for resiliency checking. While such a possibility can not be ruled out, in our proposed algorithmic framework, and also that of~\cite{chak} it does not help. 
In $m$-resiliency checking we want to check if the value
$P_1=\sum_{wt(\var{u}) \leq m} {W_f(\var{u})^2}$ is greater than $0$. 
In this direction~\cite{chak} designs a quantum algorithm so that the probability of getting a predetermined state is $\frac{P_1}{D_1}$ where $D_1=2^{2n}$ and then apply amplitude amplification on it, which gives advantage over known classical algorithms. 
Here following points are important. 
\begin{enumerate}

\item In the algorithm due to~\cite{chak}, $D_1=2^{2n}$. Thus the denominator in the 
probability expression solely depends on $n$.

\item $P_1 \leq 2^{2n}$ and this bound is tight due to Parseval's identity. Thus it is not possible to have $D_1 < 2^{2n}$ if it depends solely on $n$. 

\end{enumerate}
This is important as if the probability of getting a state is $\frac{P}{D}$ where $P$ and $D$ are both integers, then the amplitude amplification algorithm~(\cite{brass}) has the query complexity
of $\mathcal{O}(\frac{\sqrt{D}}{\sqrt{P}})$. Thus with $P_1$ fixed for resiliency checking, if one could decrease the initial value of $D_1$ using a different sampling algorithm, it would have resulted in better query and thus time complexity for resiliency checking. However as we have discussed that is not possible if $D_1$ depends only on $n$.

Similarly for checking if two functions $f, g$ are uncorrelated of degree $m$
we need to determine if $P_2=\sum_{wt(\var{u}) \leq m} {C_{f,g}(\var{u})^2}$ is greater than $0$. If we use the algorithm due to Proposition~\ref{prop:uncor} then we get a predetermined state with probability $\frac{P_2}{D_2}$ where $D_2=2^{3n}$. Now we also know that $P_2$ can be high as $2^{3n}$ and thus $D_2$ cannot be smaller than $2^{3n}$ if it solely depends on $n$. At this point we use Dicke states intelligently to make $D_2$ depend on both $n$ and $m$ and effectively run $m+1$ different algorithms $A_i, 0 \leq i \leq m$ where the probability of getting a predetermined state in $A_i$ is $\frac{P_{2,i}}{{n \choose i}2^{2n}}$ where
$P_{2,i}=\sum_{wt(\var{u})=i} C_{f,g}(\var{u})^2$. This is what allows us to have advantage in Theorem~\ref{th:dicke} for small values of $n$. If $m= \Omega (n)$ then we essentially have same time complexity due to both Proposition~\ref{prop:uncor} and Theorem~\ref{th:dicke}. A similar approach could not be applied for resiliency because the Forrelation set-up for resiliency checking or the set-up in~\cite{chak} does not give us avenue to make $D_1$ which depends on both $n$ and $m$.
The central question thus is whether such an improvement can be obtained for resiliency as well, which we highlight in the conclusion. 

\section{Conclusion}
\label{sec:6}
Forrelation is one of the central problems in the quantum paradigm.
This has been studied to demonstrate separation between the bounded error quantum model 
and the randomized classical model, starting with the seminal paper of Aaronson et al~\cite{aaron}.
The main contribution of~\cite{aaron} concentrated on designing a class of formulations 
($k$-fold Forrelation, that can be defined for any $k$ functions) 
and provide negative results related to limitation of the classical computation
for evaluating such formulations. 

In this paper we have studied the Forrelation problem to build a unified framework
towards analyzing different Boolean function spectra. First we studied a desirable 
instantiation of $2$-fold Forrelation set-up to construct promise problems 
based on bent duality. Then we move to the $3$-fold Forrelation set-up 
and use the existing techniques as well as certain modifications 
to analyze the Walsh, cross-correlation and the autocorrelation spectra of a Boolean function. 
We first show how Walsh spectrum value estimation can be improved by a constant factor in this
framework and use that to revise the resiliency checking algorithm of~\cite{chak}. 
Then we move to cross-correlation estimation and sampling algorithms using $3$-fold Forrelation.
This also immediately provides results related to autocorrelation (when both the functions used are identical).
We present several new results in this domain and also show that the results improve 
the autocorrelation estimation algorithm due to~\cite{bera} for similar levels of accuracy. Relating 
3-fold Forrelation to cross-correlation provides a new insight in this domain of research, and to the best
of our knowledge this connection could not be identified earlier.
We also exploit the concept of Dicke states to provide a more efficient algorithm when used 
for checking if two functions $f$ and $g$ are uncorrelated of degree $m$, giving us 
polynomial advantage in certain cases. 

We conclude with three important research questions in this regard.
\begin{enumerate}
\item 
Can there be a more efficient resiliency checking algorithm than the one designed in
Section~\ref{sub:sampling}?
\item
Can there be a more efficient cross-correlation sampling algorithm than the one designed in Theorem~\ref{th:cor-samp}?
\item
Can higher order Forrelation problems ($k$-fold Forrelation where $k$ is greater than $3$) 
be used to obtain more efficient estimation and sampling algorithms for Boolean function spectra?
\end{enumerate}

\end{document}